\documentclass[11pt]{article}

\usepackage{amsmath,amsthm,amssymb,bm}
\usepackage{natbib}
\usepackage{multirow}
\usepackage{booktabs}
\usepackage{graphicx}
\usepackage{verbatim} 
\usepackage{array}
\usepackage[small]{caption}
\usepackage{setspace}

\newcolumntype{R}{>{$}r<{$}}
\newcolumntype{C}{>{$}c<{$}}

\newtheorem{prop}{Proposition}

\def\var{\mathrm{Var}}
\def\cov{\mathrm{Cov}}

\def\argmax{\mathop{\mathrm{arg\,max}}}

\def\pdisc{p_{\mathrm{disc}}}
\def\pdisct{\widetilde{p}_\mathrm{disc}}

\begin{document}

\title{Modeling and Optimization of Genetic Screens via\\
RNA Interference and FACS}
\author{Yair Goldberg %\\ Dept.\ of Statistics \\ University of Haifa \\ Israel
%\thanks{
%Dept.\ of Statistics,
%University of Haifa,
%Haifa 31905, Israel; 
%ygoldberg@stat.haifa.ac.il.
%}
\and Yuval Nov 
%\thanks{
%Dept.\ of Statistics,
%University of Haifa,
%Haifa 31905, Israel; 
%yuval@stat.haifa.ac.il.
%}
}

%Dept.\ of Statistics, University of Haifa, Israel

\date{{\small Dept.\ of Statistics, University of Haifa, Israel} \\ \bigskip January 2014}

\maketitle

\begin{abstract}
We study mathematically a method for discovering which gene is related to a cell phenotype of interest.  The method is based on RNA interference --- a molecular process for gene deactivation --- and on coupling the phenotype with fluorescence in a FACS machine.  A small number of candidate genes are thus isolated, and then tested individually.  We model probabilistically this process, prove a limit theorem for its outcome, and derive operational guidelines for maximizing the probability of successful gene discovery.

\bigskip

%\noindent\textbf{Keywords:} signal detection; Lindeberg--Feller Central Limit Theorem; genetic screening; RNA interference; flow cytometry

%\bigskip

%\noindent\textbf{2010 Mathematics Subject Classification:} 92B05, 90B99

\end{abstract}

\section{Introduction}

RNA interference (RNAi) is a natural biochemical process, in which small RNA molecules deactivate (or ``silence'') genes inside cells.  Biotechnological advances of the past decade allow scientists to exploit the RNAi mechanism and deactivate specific genes of their choosing, by introducing into cells appropriately designed RNAi molecules.   RNAi technology has thus become a powerful tool that revolutionized biomedical and genetic research.  For reviews, see \citet{dykxhoorn_Lieberman:silent_rev:2005, novina_sharp:RNAi_rev:2004, mohr_ea:genomic_screening:2010}.
%Hannon (2002) ?

An important application of RNAi technology is in gene function studies whose goal is to discover which gene --- henceforth termed \emph{the target gene} --- is related to a specific cell phenotype of interest.  In such a study, the experimenter can deactivate a candidate gene using the appropriate RNAi construct, and then observe whether the phenotype is altered; when it is, a relationship between the two is established.  However, since the organisms under study typically have many thousands of genes, genome-wide RNAi experiments of this type are often too expensive and laborious to be carried out on a gene-by-gene basis.

An alternative is \emph{pooled} RNAi screens, whereby a large number of RNAi constructs of various types (i.e., corresponding to various genes) are inserted randomly into a large population of cells.  We refer to a cell with at least one construct deactivating the target gene as a \emph{target cell}.  All cells then undergo selection based on the phenotype, and the abundance of each RNAi construct type among the selected cells is measured.  If the selection favors cells \emph{not} exhibiting the phenotype (so-called negative selection), it will result in enrichment of the target cells, and hence in a relatively high count of the RNAi constructs corresponding to the target gene.  A small number of genes exhibiting high RNAi counts (say, the three genes corresponding to the three highest counts) can then be validated in separate, individual (i.e., not pooled) RNAi experiments, in the hope that the target gene is among them.  Conversely, under positive selection, some low-count genes need to be validated.

In most pooled RNAi screens, the phenotype of interest directly influences the survivability of the cells, so that the desired selection takes place automatically as a result of inserting the RNAi constructs into the cells.  When this is not the case, it is sometimes possible to couple the phenotype with fluorescence, as measured in a flow cytometry experiment \citep[e.g.,][]{bassik_ea:rapid:2009, fellmann_ea:functional_identification:2011}.  In such an experiment, the cells are processed by a fluorescence-activated cell sorter (FACS), which first excites fluorescent-labeled molecules harbored in them, and then sorts the cells into two categories according to the resulting fluorescence intensity.  The coupling means that the cells in which the target gene was deactivated (the target cells) tend to exhibit stronger fluorescence, so the entire process results in enrichment of the target cells.  The relative abundance of the various construct types among the selected cells can then be measured, and a small number of genes corresponding to the top counts can be further validated, as described above.

In this work we model probabilistically and optimize this FACS-aided pooled RNAi experiment.  The main decision point in our analysis concerns the FACS selection criterion: if too many cells pass the selection, no detectable signal for the target gene will emanate from the construct counts; if the selection is too stringent, few or no target cells will be selected, resulting again in a failure to discover the target gene.

Several studies described statistical methodologies for analyzing RNAi experiments \citep{konig_ea:prob_based:2007, birmingham_ea:statistical_methods:2009, bassik_ea:rapid:2009, rieber_ea:RNAither:2009, siebourg_ea:stability:2012, hao_ea:limited_agreement:2013}.  Our approach differs from those pursued in these studies, in that it models probabilistically the experiment starting from its fundamentals (e.g., the distribution of the number of constructs inserted to a cell, and the distribution of a cell's fluorescence intensity), rather than being data driven.  This approach allows us to study analytically how changing the FACS selection criterion affects the probability of discovering the target gene, and thus to optimize the process with respect to this probability.

% Figure ?
% Screening review: Echeverri and Perrimon, Nature 2006;  Boutros & Ahringer 2008

\section{Model and Notation}

Let $r$ be the number of genes considered.  These genes may constitute the entire genome of the organism under study, or some sizable subset thereof (e.g., all genes related to signaling pathways).  We index the genes by $i = 1,\ldots,r$, and, without loss of generality, designate the index $i = 1$ to the target gene.  Let $n$ be the number of cells sorted by  FACS, indexed by $k = 1,\ldots,n$.  In a typical experiment $r$ is in the thousands and $n$ is in the millions.

Define $N_{k,i}$ to be the number of constructs of type $i$ inserted into cell $k$, so that the target cells (those in which the target gene was deactivated) are those satisfying $N_{k,1} \geq 1$.  Let $F_k$ be the fluorescence intensity of cell $k$, and denote by $G_1$ and $G_2$ the cumulative distribution functions (CDFs) of the fluorescence intensity of the target and non-target cells, respectively.  Then,
\begin{equation}\label{eq:F_CDF}
P(F_k \leq a) = 
	\begin{cases}
		G_1(a) & N_{k,1} \geq 1,\\
		G_2(a) & N_{k,1} = 0.
	\end{cases}
\end{equation}
It is assumed that $G_1$ is larger than $G_2$ in some sense (e.g., via the usual stochastic order, whereby $G_1(a) \leq G_2(a)$ for all $a$).  Also define
\begin{equation*}
	\overline{G}_1(a) = 1 - G_1(a),	\qquad \overline{G}_2(a) = 1 - G_2(a).
\end{equation*}

The experimenter may define the FACS selection criterion either through a percentile (i.e., the selected cells are the top $t$ percents of the cells, in terms of their fluorescence intensity, for some $t \in (0,1)$) or through a fixed threshold (i.e., the selected cells are those whose fluorescence intensity exceeds some threshold $\alpha$).  The latter criterion is more tractable mathematically, so we adopt it henceforth.  The resulting \emph{construct count} corresponding to gene $i$ is therefore
\begin{equation*}
	X_i = \sum_{k=1}^n N_{k,i} I_{\{F_k > \alpha\}}, \quad i = 1,\ldots,r.
\end{equation*}
The analysis below relies on the behavior of $M_{k,i} = N_{k,i} I_{\{F_k > \alpha\}}$, which is the contribution of cell $k$ to the construct count $X_i$.

We consider two models for the process of inserting the RNAi constructs into the cells: the multinomial model, which is simpler to analyze, and the Poisson model, which is more realistic.

\subsection{The Multinomial Model}  

In the multinomial model it is assumed that each cell always admits a single construct, so that $N_k = \sum_{i=1}^r N_{k,i} = 1$ for each cell $k = 1,\ldots,n$.  The type of the construct in each cell is equally likely to be any of the $r$ possible construct types, independent of other cells.

Define $X_0 = \sum_{k=1}^n I_{\{F_k < \alpha\}}$ to be the number of cells not selected by the FACS.  Then, the joint distribution of the $X_i$ is multinomial:
\begin{equation*}
	(X_0, X_1, X_2, \ldots, X_r) \sim \mathrm{Mult}(n, p_0, p_1, p_2, \ldots, p_r),
\end{equation*}
where
\begin{equation*}
	p_0 = \frac1{r}G_1(\alpha) + \frac{r-1}{r}G_2(\alpha), \qquad	p_1 = \frac1{r}\overline{G}_1(\alpha), \qquad p_i = \frac1{r}\overline{G}_2(\alpha), \quad i \geq 2.
\end{equation*}

It is easily verified that under the multinomial model, for any power $m > 0$,
\begin{equation} \label{eq:EM}
	E[(M_{k,1})^m] = \frac1{r}\overline{G}_1(\alpha),	\qquad\quad E[(M_{k,i})^m] = \frac1{r}\overline{G}_2(\alpha), \quad i \geq 2 
\end{equation}
and
\begin{align}
	\var(M_{k,1}) &= \frac1{r}\overline{G}_1(\alpha)\bigl(1 - \frac1{r}\overline{G}_1(\alpha)\bigr) \label{eq:VarM} \\
	\var(M_{k,i}) &= \frac1{r}\overline{G}_2(\alpha)\bigl(1 - \frac1{r}\overline{G}_2(\alpha)\bigr), \qquad i \geq 2 \\
	\cov(M_{k,1}, M_{k,i}) &= -\frac1{r^2}\overline{G}_1(\alpha)\overline{G}_2(\alpha), \qquad i \geq 2 \\
	\cov(M_{k,i}, M_{k,j}) &= -\frac1{r^2}\overline{G}_2(\alpha)^2, \quad i,j \geq 2, \quad i \neq j. \label{eq:CovM}
\end{align}

\subsection{The Poisson Model}
Under the Poisson model, the process of preparing the cells with the constructs for the FACS is comprised of two steps.  In the first step, each cell admits a Poisson number of constructs, with parameter $\lambda$ that is called ``multiplicity of infection.''  The value of $\lambda$ is typically low, in the range 0.1--1, and we treat it as exogenously given, rather than as a decision variable.  As in the multinomial model, the type of each construct is assumed to be drawn uniformly from $\{1,\ldots,r\}$, independent of other constructs.  Because the  support of the Poisson distribution includes the value 0, a cell may contain no constructs, in which case it will contribute no useful information for the experiment.  To avoid this, in the second step, all cells having no constructs are eliminated, and only those with at least one construct are processed by the FACS machine.  Thus, the total number of constructs per cell has a Poisson distribution truncated below 1.  

\begin{prop}\label{prop:poisson_moments}  In the Poisson model, the contribution $M_{k,i}$ of cell $k$ to the construct count $X_i$ satisfies
\begin{align*}
	E(M_{k,1}) &= \frac{\lambda c_1}{r(1 - e^{-\lambda})} \\
	E(M_{k,i}) &= \frac{\lambda c_2}{r(1 - e^{-\lambda})}, \qquad i \geq 2 \\
	\var(M_{k,1}) &= \frac{c_1}{1 - e^{-\lambda}}\Bigl(\frac{\lambda}{r} + \frac{\lambda^2}{r^2}\Bigr) - \Bigl(\frac{\lambda c_1}{r(1 - e^{-\lambda})}\Bigr)^2 \\
	\var(M_{k,i}) &= \frac{c_2}{1 - e^{-\lambda}}\Bigl(\frac{\lambda}{r} + \frac{\lambda^2}{r^2}\Bigr) - \Bigl(\frac{\lambda c_2}{r(1 - e^{-\lambda})}\Bigr)^2, \qquad i \geq 2 \\
	\cov(M_{k,1}, M_{k,i}) &= \frac{c_1\lambda^2}{r^2(1 - e^{-\lambda})} - \frac{c_1 c_2\lambda^2}{r^2(1 - e^{-\lambda})^2} \qquad i \geq 2 \\
	\cov(M_{k,i}, M_{k,j}) &= \frac{c_2\lambda^2}{r^2(1 - e^{-\lambda})} - \frac{c_2^2\lambda^2}{r^2(1 - e^{-\lambda})^2} \qquad i,j \geq 2, \quad i \neq j \\
	E[(M_{k,1})^3] &= \frac{c_1}{1 - e^{-\lambda}}\Bigl[\frac{\lambda}{r} + 3\Bigl(\frac{\lambda}{r}\Bigr)^2 + \Bigl(\frac{\lambda}{r}\Bigr)^3\Bigr] \\
	E[(M_{k,i})^3] &= \frac{c_2}{1 - e^{-\lambda}}\Bigl[\frac{\lambda}{r} + 3\Bigl(\frac{\lambda}{r}\Bigr)^2 + \Bigl(\frac{\lambda}{r}\Bigr)^3\Bigr] \qquad i \geq 2.
\end{align*}
where
\begin{align*}
	c_1 &= c_1(\alpha) = \overline{G}_1(\alpha), \\
	c_2 &= c_2(\alpha) = \overline{G}_2(\alpha)e^{-\lambda/r} + \overline{G}_1(\alpha)(1 - e^{-\lambda/r}).
\end{align*}
\end{prop}

\begin{proof}
Because of the elimination of cells having zero constructs, the counts $N_{k,1}, \ldots, N_{k,r}$ at each cell $k$ satisfy
\begin{equation*}
(N_{k,1},\ldots,N_{k,r}) \stackrel{d}{=} (\widehat{N}_{k,1},\ldots,\widehat{N}_{k,r}) \,|\, \widehat{N}_k \geq 1, \qquad k = 1,\ldots,n,
\end{equation*}
where $\stackrel{d}{=}$ denotes equality in distribution, $\widehat{N}_{k,1}, \ldots, \widehat{N}_{k,r}$ are independent $\text{Poisson}(\lambda/r)$ random variables, and $\widehat{N}_k = \sum_{i=1}^r \widehat{N}_{k,i} \sim \mathrm{Poisson}(\lambda)$.  

The distribution of a cell's fluorescence depends only on the presence of constructs of type 1 (recall equation \eqref{eq:F_CDF}), so for $x \geq 1$ we have
\begin{equation}\label{eq:beta}
	P(F_k > \alpha \,|\, \widehat{N}_{k,1} = x) = \overline{G}_1(\alpha) = c_1,
\end{equation}
and for $i \geq 2$, since $\widehat{N}_{k,1}$ and $\widehat{N}_{k,i}$ are independent,
\begin{align}\label{eq:gamma}
	P(F_k > \alpha \,|\, \widehat{N}_{k,i} = x) &= P(F_k > \alpha \,|\, \widehat{N}_{k,i} = x, \widehat{N}_{k,1} = 0)P(\widehat{N}_{k,1} = 0 \,|\, \widehat{N}_{k,i} = x) \notag\\
	& \qquad + P(F_k > \alpha \,|\, \widehat{N}_{k,i} = x, \widehat{N}_{k,1} \geq 1)P(\widehat{N}_{k,1} \geq 1 \,|\, \widehat{N}_{k,i} = x) \notag\\
	&= P(F_k > \alpha \,|\, \widehat{N}_{k,1} = 0)P(\widehat{N}_{k,1} = 0) + P(F_k > \alpha \,|\, \widehat{N}_{k,1} \geq 1)P(\widehat{N}_{k,1} \geq 1) \notag\\
	&= \overline{G}_2(\alpha)e^{-\lambda/r} + \overline{G}_1(\alpha)(1 - e^{-\lambda/r}) \notag\\
	&= c_2.
\end{align}

Using~\eqref{eq:beta}, we have for $x \geq 1$,
\begin{align*}
	P(M_{k,1} = x) &= P(N_{k,1} = x, F_k > \alpha) \\
	&= P(\widehat{N}_{k,1} = x, F_k > \alpha \,|\, \widehat{N}_k \geq 1) \\
	&= \frac{P(\widehat{N}_{k,1} = x, F_k > \alpha)}{P(\widehat{N}_k \geq 1)} \\
	&= \frac{P(\widehat{N}_{k,1} = x)P(F_k > \alpha \,|\, \widehat{N}_{k,1} = x)}{P(\widehat{N}_k \geq 1)} \\
	&= \frac{e^{-\lambda/r}(\lambda/r)^x c_1}{x!(1 - e^{-\lambda})}.
\end{align*}
Similarly, using~\eqref{eq:gamma}, for $i,j \geq 2$, $i \neq j$ and $x,y \geq 1$ we get
\begin{align*}
	P(M_{k,i} = x) &= \frac{e^{-\lambda/r}(\lambda/r)^x c_2}{x!(1 - e^{-\lambda})} \\
	P(M_{k,1} = x, M_{k,i} = y) & = \frac{e^{-2\lambda/r}(\lambda/r)^{x+y} c_1}{x!y!(1 - e^{-\lambda})} \\
	P(M_{k,i} = x, M_{k,j} = y) & = \frac{e^{-2\lambda/r}(\lambda/r)^{x+y} c_2}{x!y!(1 - e^{-\lambda})}.
\end{align*}

Computing now the moments through their basic definitions --- e.g., $E(M_{k,1}) = \sum_{x=1}^\infty xP(M_{k,1} = x)$ --- the proposition is proved.

\end{proof}

When $\lambda$ is near zero, there is a low probability that a cell will admit two constructs or more.  Because cells with no constructs are eliminated, the probability in this case of having eventually a single construct is close to 1, similar to the multinomial model, in which there is always a single construct in each cell.  The next proposition formalizes this observation, and asserts that the entire Poisson model converges in distribution to the multinomial model as $\lambda \to 0$.  This proposition is the only place in this work in which the multinomial and the Poisson models are considered simultaneously; to distinguish between them notationally, we attach a superscript $(\lambda)$ to all random variables related to the Poisson model.

\begin{prop} Let $X_1, X_2, \ldots, X_r$ denote the construct counts under the multinomial model, and let $X^{(\lambda)}_1, X^{(\lambda)}_2, \ldots, X^{(\lambda)}_r$ denote the construct counts under the Poisson model.  Then,
\begin{equation*}
(X^{(\lambda)}_1, X^{(\lambda)}_2, \ldots, X^{(\lambda)}_r) \Rightarrow (X_1, X_2, \ldots, X_r) \qquad\text{as } \lambda \to 0.
\end{equation*}
\end{prop}

\begin{proof}
As in the proof of Proposition~\ref{prop:poisson_moments}, let $\widehat{N}_k^{(\lambda)} = \sum_{i=1}^r \widehat{N}^{(\lambda)}_{k,i} \sim \mathrm{Poisson}(\lambda)$ be the total number of constructs inserted into cell $k$ \emph{before} eliminating the empty cells, under the Poisson model. Using l'Hospital's rule, we have
\begin{align*}
	\lim_{\lambda \to 0} P(N_k^{(\lambda)} = 1) &= \lim_{\lambda \to 0} P(\widehat{N}_k^{(\lambda)} = 1 \,|\, \widehat{N}_k^{(\lambda)} \geq 1) \\
	&= \lim_{\lambda \to 0} \frac{P(\widehat{N}_k^{(\lambda)} = 1)}{P(\widehat{N}_k^{(\lambda)} \geq 1)} \\
	&= \lim_{\lambda \to 0} \frac{\lambda e^{-\lambda}}{1 - e^{-\lambda}} \\
	&= \lim_{\lambda \to 0} \frac{e^{-\lambda} - \lambda e^{-\lambda}}{e^{-\lambda}} \\
	&= 1.
\end{align*}
Thus, $N_k^{(\lambda)} \Rightarrow 1$ as $\lambda \to 0$ for each cell $k$.  The result now follows from the Continuous Mapping Theorem.
\end{proof}

\subsection{Maximizing the Probability of Discovery}

Under either the multinomial or the Poisson model, let $X_{[1]} \geq X_{[2]} \geq \ldots \geq X_{[r-1]}$ be the order statistics of the $r-1$ non-target construct counts $X_2,\ldots,X_r$, sorted from largest to smallest.  Also let $v$ denote the number of genes to be validated; this number is assumed to be exogenously given, according to budget constraints.  The target gene is discovered in the experiment if it is among the $v$ genes that are validated, an event that occurs if the construct count of the target gene is among the $v$ top counts.  Mathematically, the \emph{probability of discovery} is 
\begin{equation}\label{eq:pdisc}
\pdisc(\alpha) = P(X_{[v]} < X_1).
\end{equation}
Our goal is to find a threshold $\alpha^\ast$ that maximizes this probability, i.e., that satisfies
\begin{equation*}
\alpha^\ast = \argmax_\alpha \pdisc(\alpha).
\end{equation*}

\section{Asymptotic Analysis}

Let $n$, the number of cells, approach infinity, and assume that the number of genes grows to infinity with $n$, i.e., $r = r(n) \to \infty$ as $n \to \infty$.  We attach a superscript $n$ to all random variables defined above, so that $X^n_i = \sum_{k=1}^n M^n_{k,i}$ is the construct count corresponding to gene $i$ in the $n$th system.  Let
\begin{equation}\label{eq:Yn}
	Y_i^n = \frac{X_i^n - E(X_i^n)}{\sqrt{\var(X_i^n)}},  \qquad i = 1,2,\ldots,r(n),
\end{equation}
be the normalized construct counts, and define the process
\begin{equation}\label{eq:Yn_vector}
	\mathbf{Y}^n = (Y_1^n, Y_2^n, \ldots, Y_{r(n)}^n, 0, 0, \ldots).
\end{equation}

The following result asserts that if $r(n)$ approaches infinity slower than $n$, then the scaled construct counts are asymptotically normal and independent.

\begin{prop}\label{prop:normal_approx}
Under both the multinomial and the Poisson models, and for fixed $\alpha$, if $r(n) \to \infty$ and $n/r(n) \rightarrow \infty$ as $n \to \infty$, then 
\begin{equation*}
	\mathbf{Y}^n \Rightarrow (Z_1, Z_2,\ldots) \quad\text{as } n \to \infty,
\end{equation*}
where the $Z_i$ are independent standard normal random variables.
\end{prop}

\begin{proof}
By theorem 1.4.8 of \citet{vandervaart_wellner:weak_conv:1996}, it is enough to prove finite-dimensional convergence, i.e., to show that for each $d \in \mathbb{N}$,
\begin{equation*}
	Y^n = (Y_1^n, \ldots, Y_d^n)^T \Rightarrow (Z_1,\ldots, Z_d)^T.
\end{equation*}
By the Cram\'{e}r--Wold theorem, it needs to be shown that for each $a \in \mathbb{R}^d$,
\begin{equation}\label{eq:d_univariate}
	a^TY^n \Rightarrow N(0, a^Ta).
\end{equation}

Define
\begin{equation*}
	Q_{k,i}^n = \frac{M_{k,i}^n - E(M_{k,i}^n)}{\sqrt{n\var(M_{k,i}^n)}}
\end{equation*}
Then,
\begin{align}
	E(Q_{k,i}^n) &= 0 \label{eq:EQ} \\
	\var(Q_{k,i}^n) &= 1/n \\
	\cov(Q_{k,1}^n, Q_{k,i}^n) &= \frac{\cov(M_{k,1}^n, M_{k,2}^n)}{n\sqrt{\var(M_{k,1}^n)\var(M_{k,2}^n)}}, \qquad &i \geq 2 \\
	\cov(Q_{k,i}^n, Q_{k,j}^n) &= \frac{\cov(M_{k,2}^n, M_{k,3}^n)}{n\var(M_{k,2}^n)},  & i,j \geq 2, \quad i \neq j. \label{eq:CovQ}
\end{align}

Define $V_k^n = \sum_{i=1}^d a_i Q^n_{k,i}$.  Then,
\begin{align*}
	\sum_{k=1}^n V_k^n &= \sum_{k=1}^n\sum_{i=1}^d a_i Q_{k,i}^n \\
	& = \sum_{i=1}^d a_i \sum_{k=1}^n \frac{M_{k,i}^n - E(M_{k,i}^n)}{\sqrt{n\var(M_{k,i}^n)}} \\
	& = \sum_{i=1}^d a_i \frac{X_i^n - nE(M_{k,i}^n)}{\sqrt{n\var(M_{k,i}^n)}} \\
	& = a^TY^n.
\end{align*}
Thus, proving \eqref{eq:d_univariate} is equivalent to proving $\sum_{k=1}^n V_k^n \Rightarrow N(0, a^Ta)$.  Now consider the triangular array $\{V_k^n,\ k = 1,\ldots,n,\ n = 1, 2, \ldots\}$.  Using \eqref{eq:EQ}--\eqref{eq:CovQ}, we have that $E(V_k^n) = 0$ and
\begin{align*}
	\var(V_k^n) &= \sum_{i=1}^d \var(a_i Q^n_{k,i}) + \sum_{i\neq j}\cov(a_i Q^n_{k,i}, a_j Q^n_{k,j}) \\
	&= \frac1{n}\sum_{i=1}^d a_i^2 
		+ \frac{2}{n} \Biggl[\frac{\cov(M_{k,1}^n, M_{k,2}^n)}{\sqrt{\var(M_{k,1}^n)\var(M_{k,2}^n)}}\sum_{i=2}^d a_1 a_i 
		+ \frac{\cov(M_{k,2}^n, M_{k,3}^n)}{\var(M_{k,2}^n)} \sum_{\substack{i,j \geq 2 \\ i < j}} a_i a_j \Biggr].
\end{align*}
Therefore,
\begin{align}
	s^2_n &= \var\biggl(\sum_{k=1}^n V_k^n\biggr) \notag \\
		&= n\var(V_1^n) \notag \\
		&= \sum_{i=1}^d a_i^2 
		+ 2\Biggl[\frac{\cov(M_{k,1}^n, M_{k,2}^n)}{\sqrt{\var(M_{k,1}^n)\var(M_{k,2}^n)}}\sum_{i=2}^d a_1 a_i 
		+ \frac{\cov(M_{k,2}^n, M_{k,3}^n)}{\var(M_{k,2}^n)} \sum_{\substack{i, j \geq 2 \\ i < j}} a_i a_j \Biggr]. \label{eq:s^2}
\end{align}

Using equations \eqref{eq:VarM}--\eqref{eq:CovM} for the multinomial model, or Proposition~\ref{prop:poisson_moments} for the Poisson model, we have that the coefficients before both sums in the square brackets in the last expression converge to zero as $n \rightarrow \infty$, as the numerator in each is $O(1/r^2)$, and the denominator is $O(1/r)$.  Therefore, $s^2_n \rightarrow \sum_{i=1}^d a_i^2$ as $n \rightarrow \infty$.

\medskip

By the Lindeberg--Feller Central Limit Theorem, a sufficient condition for 
\begin{equation}\label{eq:LF}
	\frac1{s_n}\sum_{k=1}^n V_k^n \Rightarrow N(0,1)
\end{equation}
is the Lyapunov condition:
\begin{equation*}
	\text{there exists $\delta > 0$ such that} \quad \frac1{s_n^{2 + \delta}}\sum_{k=1}^n E\bigl(|V_k^n|^{2+\delta} \bigr) \rightarrow 0 \quad\text{as } n \rightarrow \infty.
\end{equation*}
We now show that the condition holds for $\delta = 1$.  Since $V_1^n,\ldots,V_n^n$ are identically distributed for each $n$, the Lyapunov condition in the case $\delta = 1$ reduces to
\begin{equation*}
 \frac{n}{s_n^3}E\bigl(|V_1^n|^3\bigr) \rightarrow 0 \quad\text{as } n \rightarrow \infty.
\end{equation*}

Using Minkowski inequality and the fact that the $M_{k,i}^n$ are non-negative, we have that
\begin{align*}
 \frac{n}{s_n^3}E\bigl(|V_1^n|^3\bigr) 
	& = \frac{n}{s_n^3}E\Biggl(\Biggl|\sum_{i=1}^d a_i \frac{M_{1,i}^n - E(M_{1,i}^n)}{\sqrt{n\var(M_{1,i}^n)}}\Biggr|^3\Biggr) \\
	& = \frac1{s_n^3} E\Biggl(\Biggl|\sum_{i=1}^d \frac{a_i[M_{1,i}^n - E(M_{1,i}^n)]}{n^{1/6}\sqrt{\var(M_{1,i}^n)}}\Biggr|^3\Biggr) \\
	& \leq \frac1{s_n^3} \Bigg[\sum_{i=1}^d \Biggl(E\Biggl|\frac{a_i[M_{1,i}^n - E(M_{1,i}^n)]}{n^{1/6}\sqrt{\var(M_{1,i}^n)}}\Biggr|^3\Biggr)^{\!1/3}\, \Bigg]^3 \\
%	& \leq \frac1{s_n^3} E\Biggl(\Biggl|\sum_{i=1}^d \frac{a_i M_{1,i}^n}{n^{1/6}\sqrt{\var(M_{1,i}^n)}}\Biggr|^3\Biggr) \\
	& \leq \frac1{s_n^3} \Bigg[\sum_{i=1}^d \Biggl(E\Biggl|\frac{a_i M_{1,i}^n}{n^{1/6}\sqrt{\var(M_{1,i}^n)}}\Biggr|^3\Biggr)^{\!1/3}\, \Bigg]^3.
\end{align*}

Recall that $s_n^2$ converges to a constant.  Thus, since the sum in the last expression involves a finite and fixed number of summands, to show that the last expression converges to zero, it is enough to show that for each $i$, the expression
\begin{equation}\label{eq:summand}
 E\Biggl|\frac{a_i M_{1,i}^n}{n^{1/6}\sqrt{\var(M_{1,i}^n)}}\Biggr|^3 = \frac{|a_i^3|E[(M_{1,i}^n)^3]}{n^{1/2}[\var(M_{1,i}^n)]^{3/2}} \\
\end{equation}
converges to zero.  Indeed, using equations~\eqref{eq:EM}--\eqref{eq:CovM} for the multinomial model, or Proposition~\ref{prop:poisson_moments} for the Poisson model, we have that $E[(M_{1,i}^n)^3]$ converges to zero at rate $1/r$, whereas $[\var(M_{1,i}^n)]^{3/2}$ does so at rate $1/r^{3/2}$.  Thus, the entire right-hand side of the last displayed equation is of order $(r/n)^{1/2}$, and since we assumed that $n/r \rightarrow \infty$, the Lyapunov condition is satisfied.

\medskip

We have shown that \eqref{eq:LF} holds.  What we need to show is \eqref{eq:d_univariate}, which may be written equivalently as
\begin{equation*}
	\frac1{\|a\|_2}\sum_{k=1}^n V_k^n \Rightarrow N(0,1).
\end{equation*}
However, since $s_n \rightarrow (\sum_{i=1}^d a_i^2)^{1/2} = \|a\|_2$, the proposition is proved.
\end{proof}

\section{Approximating the Probability of Discovery}

Under either the multinomial or the Poisson model, evaluating the exact probability of discovery $\pdisc(\alpha)$ in~\eqref{eq:pdisc} is difficult, because of the dependence among the $X_i$.  We therefore use the asymptotic result of the previous section to derive an approximation to $\pdisc(\alpha)$.  For fixed $n$ and $r$, let $\widetilde{X}_1,\ldots,\widetilde{X}_r$ be independent normal random variables, with $E(\widetilde{X}_i) = E(X_i)$ and $\var(\widetilde{X}_i) = \var(X_i)$.  By Proposition~\ref{prop:normal_approx}, when both $n$ and $r$ are large, but $r \ll n$, these $\widetilde{X}_1,\ldots,\widetilde{X}_r$ may serve as approximations to $X_1,\ldots,X_r$.  Let $\phi_1$ be the density function of $\widetilde{X}_1$, and $\Phi_2$ be the CDF of $\widetilde{X}_i$ for $i \geq 2$ (recall that $X_2,\ldots,X_r$ are identically distributed); note that both $\phi_1$ and $\Phi_2$ depend on $\alpha$.  Also let $\widetilde{X}_{[1]} \geq \widetilde{X}_{[2]} \geq \ldots \geq \widetilde{X}_{[r-1]}$ be the order statistics of $\widetilde{X}_2, \ldots, \widetilde{X}_r$, and define 
\begin{equation}\label{eq:pdisct}
	\pdisct(\alpha) = P(\widetilde{X}_{[v]} < \widetilde{X}_1)
\end{equation}
to be the approximation to the probability of discovery $\pdisc(\alpha)$ in~\eqref{eq:pdisc}.
\begin{prop}
\begin{equation*}
\pdisct(\alpha) = \int_{-\infty}^\infty \sum_{j=r-v}^{r-1} \binom{r-1}{j}[\Phi_2(x)]^j [1 - \Phi_2(x)]^{r-j-1} \phi_1(x)\, dx.
\end{equation*}
\end{prop}

\begin{proof}
Because $\widetilde{X}_2,\ldots,\widetilde{X}_r$ are iid with CDF $\Phi_2$, the CDF of $\widetilde{X}_{[v]}$ is
\begin{equation*}
P(\widetilde{X}_{[v]} \leq x) = \sum_{j=r-v}^{r-1} \binom{r-1}{j}[\Phi_2(x)]^j [1 - \Phi_2(x)]^{r-j-1}.
\end{equation*}
See p.\ 87 of \citet{serfling:approx_thm:1981}.  Conditioning on the value of $\widetilde{X}_1$ in the right-hand side of~\eqref{eq:pdisct} and integrating with respect to its density $\phi_1$, the proposition is proved.
\end{proof}

The dashed curve in Figure~\ref{fig:alpha_vs_detect} shows $\pdisct(\alpha)$, the approximate probability of discovery, as a function of $\alpha$. The solid curve is the true probability of discovery, $\pdisc(\alpha)$, as estimated by simulation.  The system parameters are $r = 200$ genes, $n = 40,000$ cells, $v = 3$ genes to be validated, fluorescence distributions $G_1 = N(0.4, 1)$ and $G_2 = N(0, 1)$, and the multinomial model. 

\begin{figure}[htp]
\centering
\includegraphics[width=0.5\textwidth]{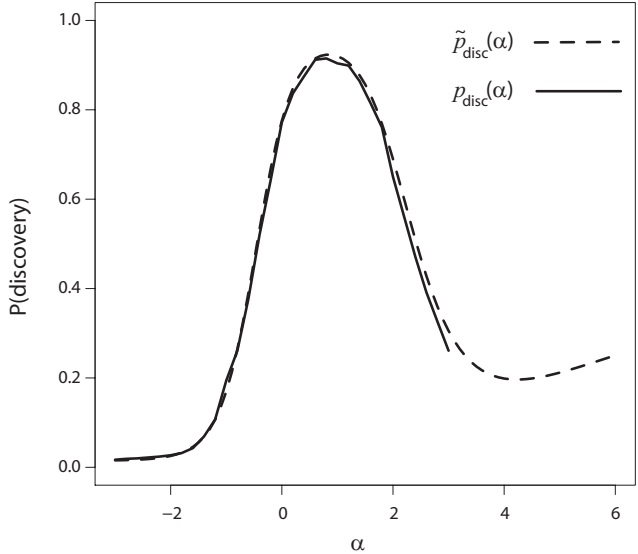}
\caption{The probability of discovering the target gene as a function of the selection threshold $\alpha$.  Solid curve is the true probability of discovery for the exact system, $\pdisc(\alpha)$, as estimated by simulation.  The dashed curve is the approximate probability of discovery, $\pdisct(\alpha)$.}\label{fig:alpha_vs_detect}
\end{figure}

The main feature of Figure~\ref{fig:alpha_vs_detect} is that the approximate curve follows the exact curve very closely.  Thus, the asymptotics-based approximation works well in practice.  The optimal threshold $\alpha^\ast$ for this system parameters is about $0.8$.  Note also that the curve of $\pdisc(\alpha)$ is plotted only for $\alpha \leq 3$; the reason for this is that for $\alpha > 3$ the selection criterion is too stringent, so that in practice no cells satisfy it, and the construct counts --- which are always integer in the exact system --- are all zero.  In contrast, the counts in the approximate system are continuous random variables, which may assume near-zero values, and so $\pdisct(\alpha)$ can be computed for any $\alpha$.

Figure~\ref{fig:other_params} shows how the curve $\pdisct(\alpha)$ changes with the system parameters.  The left panel shows the influence of the separation between the two fluorescence distributions $G_1$ and $G_2$, and the right panel the influence of $v$, the number of genes to be validated.  As expected, better separation and higher $v$ result in higher probabilities of discovery.

\begin{figure}[htp]
\centering
\includegraphics[width=\textwidth]{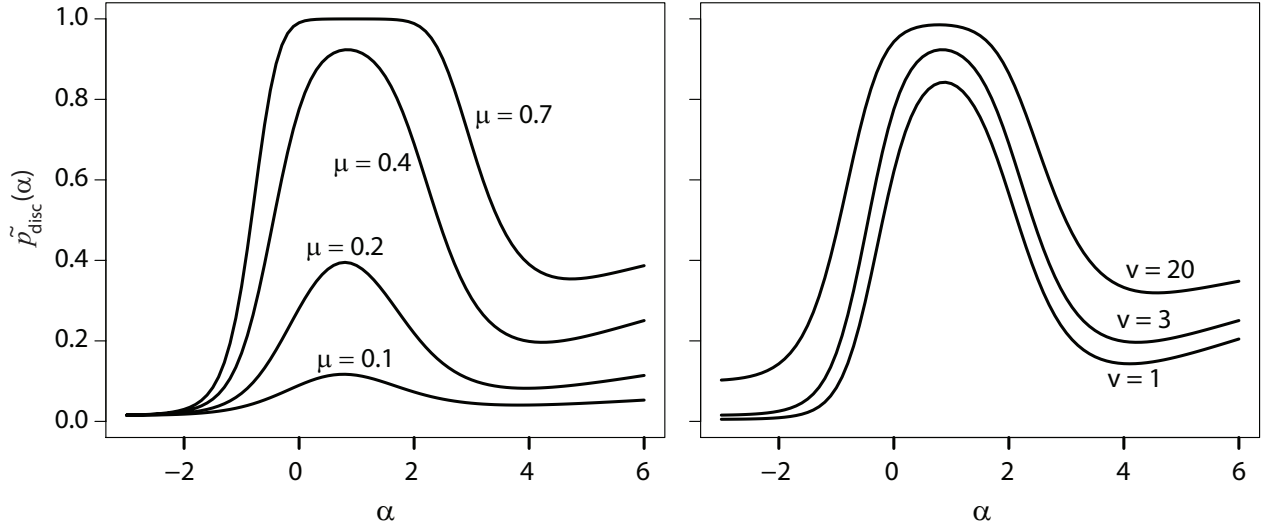}
\caption{The approximate probability of discovery $\pdisct(\alpha)$ as a function of $\alpha$, for various system parameters.  Left panel: $G_1 = N(\mu, 1)$ for various mean values $\mu$, and $G_2 = N(0, 1)$.  Right panel: various values of $v$, the number of genes to be validated.}\label{fig:other_params}
\end{figure}

\section{The Two-Stage Discovery Problem}\label{sec:2-stage}

After enriching the target cells by FACS, it is possible to enrich them further, by first growing the selected cells until their population is large enough, and then processing them in a second FACS round.

We model this two-stage process as follow.  We let $n$ be the number of cells processed by FACS in the first stage.  As before, $N_{k,i}$ denotes the number of constructs of type $i$ inserted into cell $k$ (according to either the multinomial model or the Poisson model), and $F_k$ denotes the fluorescence intensity of cell $k$, which is determined according to~\eqref{eq:F_CDF}.  The selection criterion for cell $k$ in the first stage is $F_k > \alpha$, for some threshold $\alpha$.  Each selected cell from the first stage gives rise to $L$ descendant cells, to be processed in the second stage; all $L$ descendants of the same ancestor cell inherit the RNAi construct content of their ancestor.  The fluorescence intensity of the $l$th descendant of cell $k$ is measured in the second FACS stage, and is denoted by $F_{k,l}$; the distribution of $F_{k,l}$ is the same as that of the ancestor cell, i.e.,
\begin{equation*}
P(F_{k,l} \leq a) = 
	\begin{cases}
		G_1(a) & N_{k,1} \geq 1,\\
		G_2(a) & N_{k,1} = 0.
	\end{cases}
\end{equation*}
For each ancestor cell $k$, the $L$ fluorescence intensities $F_{k,1}, \ldots, F_{k,L}$ of the $L$ descendant cells are conditionally independent given $N_{k,1}$.  The selected cells in the second stage are those satisfying $F_{k,l} > \beta$, for some threshold $\beta$.  The final construct counts are given by
\begin{equation*}
X_i = \sum_{k=1}^n \sum_{l=1}^L N_{k,i} I_{\{F_k > \alpha,\ F_{k,l} > \beta\}}, \qquad i = 1,\ldots,r.
\end{equation*}

Let $T_{k,i} = N_{k,i}I_{\{F_k > \alpha\}} \sum_{l=1}^L I_{\{F_{k,l} > \beta\}}$, so that $X_i = \sum_{k=1}^n T_{k,i}$.  The $T_{k,i}$ are thus the counterparts of the $M_{k,i}$ from the single-stage problem.

\begin{prop}\label{prop:T_moments}
Under the multinomial model, the $T_{k,i}$ satisfy
\begin{align*}
	E(T_{k,1})& = \frac1{r}L\overline{G}_1(\alpha)\overline{G}_1(\beta) \\
	E(T_{k,i})& = \frac1{r}L\overline{G}_2(\alpha)\overline{G}_2(\beta) \qquad i \geq 2 \\
	\var(T_{k,1}) &= \frac1{r}\overline{G}_1(\alpha)\bigl[LG_1(\beta)\overline{G}_1(\beta) + L^2\overline{G}_1(\beta)^2\bigr] - \frac1{r^2}\overline{G}_1(\alpha)^2 L^2\overline{G}_1(\beta) \\
	\var(T_{k,i}) &= \frac1{r}\overline{G}_2(\alpha)\bigl[LG_2(\beta)\overline{G}_2(\beta) + L^2\overline{G}_2(\beta)^2\bigr] - \frac1{r^2}\overline{G}_2(\alpha)^2 L^2\overline{G}_2(\beta) \qquad i \geq 2 \\
	\cov(T_{k,1}, T_{k,i}) &= -\frac1{r^2} L^2 \overline{G}_1(\alpha)\overline{G}_1(\beta)\overline{G}_2(\alpha)\overline{G}_2(\beta) \qquad i \geq 2 \\
	\cov(T_{k,i}, T_{k,j}) &= -\frac1{r^2} L^2 \overline{G}_2(\alpha)^2\overline{G}_2(\beta)^2 \qquad i,j \geq 2, \quad i\neq j\\
	E[(T_{k,1})^3] &= \frac1{r}L\overline{G}_1(\alpha)\overline{G}_1(\beta)\bigl(L^2\overline{G}_1(\beta)^2 - 3L\overline{G}_1(\beta)^2 + 2\overline{G}_1(\beta)^2 + 3L\overline{G}_1(\beta) - 3\overline{G}_1(\beta) + 1\bigr) \\
	E[(T_{k,i})^3] &= \frac1{r}L\overline{G}_2(\alpha)\overline{G}_2(\beta)\bigl(L^2\overline{G}_2(\beta)^2 - 3L\overline{G}_2(\beta)^2 + 2\overline{G}_2(\beta)^2 + 3L\overline{G}_2(\beta) - 3\overline{G}_2(\beta) + 1\bigr) \quad i \geq 2\\
\end{align*}
\end{prop}

\begin{proof}
Under the multinomial model, $T_{k,i}$ is binomial conditional on $N_{k,i}I_{\{F_k > \alpha\}} = 1$, and zero otherwise:
\begin{align*}
T_{k,1}\,|\,N_{k,1}I_{\{F_k > \alpha\}} &= 1 \ \sim \ \textrm{Bin}(L, \overline{G}_1(\beta)) \\ 
T_{k,i}\,|\,N_{k,i}I_{\{F_k > \alpha\}} &= 1 \ \sim \ \textrm{Bin}(L, \overline{G}_2(\beta)) \quad i \geq 2.
\end{align*}
The moments of $T_{k,i}$ are then the well known moments of the binomial distribution, multiplied either by $P(N_{k,1}I_{\{F_k > \alpha\}} = 1) = \overline{G}_1(\alpha)/r$ for $i = 1$, or by $P(N_{k,i}I_{\{F_k > \alpha\}} = 1) = \overline{G}_2(\alpha)/r$ for $i \geq 2$.
\end{proof}

As in the single-stage problem, we let both $n$ and $r = r(n)$ approach infinity, and define the normalized construct count $Y_i^n$ through~\eqref{eq:Yn}, and the process $\mathbf{Y}^n$ through~\eqref{eq:Yn_vector}.  The following result is the two-stage counterpart of Proposition~\ref{prop:normal_approx}, and asserts that the scaled construct counts are again asymptotically normal and independent.

\begin{prop}\label{prop:normal_approx_2}
Under both the multinomial and the Poisson models, and for fixed thresholds $\alpha$ and $\beta$, if $r(n) \to \infty$ and $n/r(n) \rightarrow \infty$ as $n \to \infty$, then 
\begin{equation*}
	\mathbf{Y}^n \Rightarrow (Z_1, Z_2,\ldots) \quad\text{as } n \to \infty,
\end{equation*}
where the $Z_i$ are independent standard normal random variables.
\end{prop}

\begin{proof}
For brevity, we prove the proposition only for the multinomial model.  The proof follows the exact same steps as that of Proposition~\ref{prop:normal_approx}, with the $T_{k,i}$ replacing the $M_{k,i}$.  Only two points need to be reestablished:  The first is that the coefficients before both sums in the square brackets in equation~\eqref{eq:s^2} converge to zero as $n \rightarrow \infty$; this is true since by Proposition~\ref{prop:T_moments}, we again have that the numerator of each is $O(1/r^2)$, whereas the denominator is $O(1/r)$.  The second point is that for each $i$, the expression at the right-hand side of equation~\eqref{eq:summand} converge to zero as $n \rightarrow \infty$; this is again true since by Proposition~\ref{prop:T_moments}, that expression is $O((r/n)^{1/2})$, and by assumption, $n/r \rightarrow \infty$.
\end{proof}

The decision variables in the two-stage model are the thresholds $\alpha$ and $\beta$.  Clearly, it is desirable to enrich the first-stage selected cells as much as possible, and this can be done by raising $\alpha$.  However, as in the single-stage problem, setting $\alpha$ too high may result in no target cells (and hence no constructs of type 1) among the selected cells.  We resolve this conflict by maximizing $\alpha$ subject to a constraint that ensures that the number of selected target cells is high enough.  Let 
\begin{equation*}
	W_1 = W_1 (\alpha) = \sum_{k=1}^n I_{\{F_k > \alpha,\ N_{k,1} \geq 1\}}
\end{equation*}
be the number of target cells selected by the FACS.  Under the multinomial model, for example, $W_1 \sim \textrm{Bin}(n , \overline{G}_1(\alpha)/r)$.  We may then set $\alpha$ to be maximal subject to $E(W_1) \geq b$, or to $P(W_1 \geq b) \geq 1 - \epsilon$, for some $b$ and $\epsilon$.

As in the single-stage problem, we let $\pdisct(\alpha, \beta)$ be the approximate probability of discovery, based on the normal approximation from Proposition~\ref{prop:normal_approx_2}.  The solid curve in Figure~\ref{fig:detect_vs_alpha_single_vs_two_stages} depicts $\pdisct(\alpha^\ast, \beta)$ as a function of $\beta$, where $\alpha^* = 0.55$ is the maximal $\alpha$ satisfying $E(W_1) \geq 10$, and for a multinomial system with parameters $r = 200$, $n = 5000$, $v = 3$, $G_1 = N(0.3, 1)$, and $G_2 = N(0,1)$.  The parameter $L$ was set to 4, the value required so that the expected number of cells processed in the second stage is roughly 5000.  The dashed curve is $\pdisct(\alpha)$ as a function of $\alpha$ for a single-stage system with the same parameters, except for $n = 10000$ (so that the total number of cells processed by FACS in the two systems is roughly same).  Dividing the screening between two stages improves significantly the probability of discovering the target gene: the maximal probability of discovery in the two-stage system is 0.64 (achieved by $\beta^\ast = 0.1$), whereas in the single-stage system, it is 0.28 (achieved by $\alpha^\ast = 0.9$).
\begin{figure}[htp]
\centering
\includegraphics[width=0.5\textwidth]{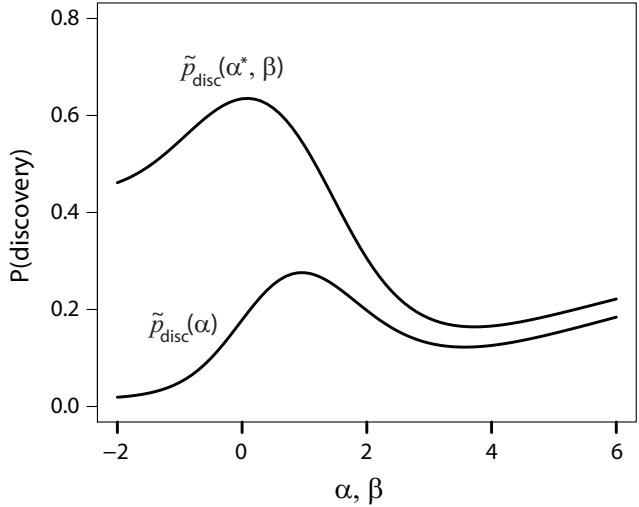}
\caption{The probability of discovery as a function of $\alpha$ in a single-stage system (lower curve) and as a function of $\beta$ in a two-stage system (upper curve).}\label{fig:detect_vs_alpha_single_vs_two_stages}
\end{figure}

\section{Discussion}

In this paper we modeled and analyzed probabilistically FACS-based RNAi genetic screening experiments.  The key decision variable in the analysis is the FACS selection threshold $\alpha$, which needs to be set optimally so as to maximize the probability of discovering the target gene.  This probability of discovery is determined by two factors: the number of the selected cells, and the enrichment level (the proportion of the target cells among the selected cells).  The strong law of large numbers guarantees that when the enrichment level is fixed, the probability of discovery approaches 1 as the number of the selected cells increases; clearly, when the number of selected cells is fixed, increasing the enrichment level also results in a higher probability of discovery.  Raising $\alpha$, therefore, has two contradicting effects on the probability of discovery, as it both decreases the number of selected cells, and increases the enrichment level.  The optimal $\alpha^\ast$ balances these opposing requirements, and can be determined through our normal approximation.

The two fluorescence distributions $G_1$ and $G_2$ were not assumed to be of any specific type in our analysis.  \citet{furusawa_ea:ubiquity:2005} advocate using a log-normal distribution to model FACS fluorescence readings.  However, since the FACS selection process is ordinal, the entire analysis is invariant under monotonically increasing transformations of the fluorescence distributions.  Log-normal distributions may thus be converted to normal ones, as we used in our simulations.

In Section~\ref{sec:2-stage} of this paper we studied a two-stage version of the discovery problem.  In principle, it is possible to repeat the enrichment--reproduction process multiple times, rather than just two, to increase further the probability of discovery.  However, each such repetition increases the likelihood of introducing a contamination into the cell population, in which case the entire experiment is lost.  We follow therefore~\citet{bassik_ea:rapid:2009}, and study only the single- and two-stage versions of the problem.

This paper is concerned with the stochastic modeling and analysis of RNAi experiments, and the statistical aspects of the problem are beyond its scope.  These aspects, however, deserve study: for example, the uncertainty resulting from estimating $G_1$ and $G_2$ can be accounted for in a more detailed analysis, and so is the noise inherent to measuring the construct counts.  We plan to study such statistical aspects, in conjunction with the above model, in a sequel to this work.

\section{Acknowledgments}
We thank Ben-Zion Levi, Nitsan Fourier, and Ofer Barnea-Yizhar from the Faculty of Biotechnology and Food Engineering at the Technion, Israel, for introducing us to the research problem and for useful discussions.  The research of Yair Goldberg was supported in part by the Israeli Science Foundation grant 1308/12.  The research of Yuval Nov was supported in part by the Israeli Science Foundation grant 286/13.

\bibliographystyle{natbib}
%\bibliography{FACS}{}

\end{document}